\documentclass[11pt]{article}

\usepackage{graphicx} 
\usepackage{dsfont}
\usepackage{epstopdf}
\usepackage{mathrsfs}
\usepackage{amsmath} 
\usepackage{amssymb}  
\usepackage{amsfonts,amsthm}
\usepackage{xcolor}
\usepackage{slashbox}
\usepackage{a4wide,epic}

\graphicspath{{./Figures/}}

\newtheorem{theorem}{Theorem}

\title{\LARGE \bf Is entanglement necessary in the reservoir input?
}

\author{Zibo~Miao,~Yu Chen~and Haidong~Yuan
	\thanks{Z. Miao, Y. Chen and H. Yuan are with the Department of Mechanical and Automation Engineering, The Chinese University of Hong Kong, Shatin, Hong Kong.  Email: {\tt\small zibomiao@cuhk.edu.hk, anschen@link.cuhk.edu.hk, hdyuan@mae.cuhk.edu.hk}}
}

\author{Zibo Miao\footnote{Z. Miao is with the School of Mechanical Engineering and Automation, Harbin Institute of Technology, Shenzhen, China. He was with the Department of Mechanical and Automation Engineering, The Chinese University of Hong Kong, Shatin, Hong Kong (Email: {\tt\small shenwum@gmail.com}).  $^\dagger$Y. Chen and H. Yuan are with the Department of Mechanical and Automation Engineering, The Chinese University of Hong Kong, Shatin, Hong Kong (Email: {\tt\small anschen@link.cuhk.edu.hk; hdyuan@mae.cuhk.edu.hk}).}
,~Yu Chen$^\dagger$~and Haidong~Yuan$^\dagger$}

\date{}
\begin{document}
	
	\maketitle
	
	\begin{abstract}
		In this paper, we continue our investigation on controlling the state of a quantum harmonic oscillator, by coupling it to a reservoir composed of a sequence of qubits. Specifically, we show that sending qubits separable from each other but initialised at different states in pairs can stabilise the oscillator at squeezed states. However, only if entanglement is allowed in the reservoir qubit can we stabilise the oscillator at a wider set of squeezed states. This thus provides a proof for the necessity of involving entanglement in the reservoir qubits input to the oscillator, as regard to the stabilisation of quantum states in the proposed system setting. On the other hand, this system setup can be in turn used to estimate the coupling strength between the oscillator and reservoir qubits. We further demonstrate that entanglement in the reservoir input qubits contributes to the corresponding quantum Fisher information. From this point of view, entanglement is proved to play an indispensable role in the improvement of estimation precision in quantum metrology.
	\end{abstract}

	\section{INTRODUCTION}
	Stabilisation of a quantum system at a desired target state plays a central role in engineering, as it dovetails with various quantum technologies such as quantum sensing \cite{CMMRD01,HR06B,LTP15}. However, in view of the short dynamical time scales, instantaneous output signal analysis and retroaction is inevitably limited in the vast majority of quantum systems. Quantum reservoir engineering, considered as an alternative control approach for quantum state stabilisation, helps us avoid a direct real-time analysis of output signals. In more concrete terms, a reservoir is designed coupled to a quantum system, with the aim of steering the system initialised at arbitrary states to a single target state or a subspace of desired states \cite{ZWH03,WHZ81,SRBR11}.
	
	In this paper, we continue our investigation on the system setup depicted in Fig. \ref{fig:GREng}, in line with our previous studies detailed in \cite{SRBR11,SLBRR12,MS17}. There is a single harmonic oscillator mode of a cavity that is weakly coupled to a sequence of qubits, via the Jaynes-Cummings Hamiltonian. This setup is generally analogous to the Haroche experimental setting \cite{SDZ11}, but no measurement is imposed. Each qubit can be initialised as needed before entering the cavity,  and then it interacts with the stored oscillator mode before exiting the cavity. Finally it will be discarded when the next qubit moves into the cavity.  The stream of qubits thus acts as an engineered reservoir to control the oscillator's quantum state. It is known that a short resonant interaction with a stream of independent, identical, weakly excited qubits can stabilise coherent states inside the cavity, whereas entangled qubits can stabilise e.g.~ squeezed states of the field \cite{SRBR11,SLBRR12,MS17}.
	
	As a further step towards better understanding the necessity of involving entanglement in reservoir input qubits, here we take into account the scenario where the input qubits are initialised at different states, acting as a time-varying quantum reservoir. We explore the beneficial effects of time-varying features, in comparison with embracing entanglement in the reservoir qubits. In particular, we observe that by alternately sending separable reservoir qubits initialised at two different states to the oscillator, squeezed states can be stabilised, while entangled reservoir input qubits can stabilise the oscillator mode at a wider set of squeezed states. These results provide important insight on the necessity of having entanglement in reservoir input.  
	
	\begin{figure}[!htp]
		\centering
		\includegraphics[scale=.6]{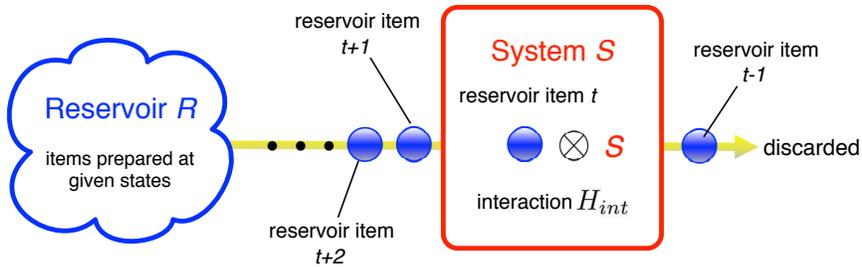}
		\caption{Framework of open-loop quantum reservoir engineering. The aim is to stabilise the system S at a target state by coupling it to another quantum system R, referred to as the reservoir, which is viewed as a stream of input quantum states that are discarded after interaction.}
		\label{fig:GREng}
	\end{figure}
	
	On the other hand, in quantum metrology this system setup can be utilised as an apparatus to estimate the coupling strength between reservoir qubits and the oscillator. Quantum metrology, studying how to obtain higher statistical precision than purely classical approaches by taking advantage of quantum resources, has recently raised much attention.  
	The precision limit of estimating a parameter $x$, encoded in a quantum state $\rho_{x}$ of the system, can be calibrated by the quantum Cram\'er-Rao bound $\delta \hat{x} \ge \frac{1}{\sqrt{nJ_Q}}$ \cite{Holevo82B,Helstrom76B,BC94B}, where $\delta \hat{x}=\sqrt{E[(\hat{x}-x)^2]}$ is the standard deviation of an unbiased estimator $\hat{x}$, and $n$ is the number of repeated experiments. Here $J_Q = \mathrm{tr}[\rho_xL_x^2]$ denotes the quantum Fisher information (QFI) with $L_x$ being the symmetric logarithmic derivative (SLD) defined as the solution to the equation $\frac{\partial \rho_x}{\partial x} = \frac{L_x\rho_x+\rho_xL_x}{2}$. The QFI thus bounds the precision limit, and larger QFI guarantees higher precision. We find that entanglement in reservoir input qubits plays an indispensable role in the improvement of estimation precision. In this regard, the necessity of incorporating entanglement in the reservoir is also explained.

	This paper is organised as follows. We begin in Section \ref{sec:RQIO} by  presenting the mathematical description for reservoir qubits in pairs input to an oscillator mode consecutively. In Section \ref{sec:NIE}, we discuss the necessity of incorporating entanglement in the reservoir qubits from two different perspectives: stabilisation of squeezed states, and enhancement of parameter-estimation. Finally, Section \ref{sec:C} provides some concluding remarks and future research directions.

	{\em Notations}. In this paper $\ast$ is used to indicate the complex conjugate $z^{\ast}=x-iy$ of a complex number $z=x+iy$ (here, $i=\sqrt{-1}$ and $x,y$ are real). Real and imaginary parts are denoted by $\Re\left(z\right)=\frac{z+z^{\ast}}{2}$ and $\Im\left(z\right)=\frac{z-z^{\ast}}{2i}$ respectively.  The Hilbert space adjoin of an operator $\mathbf{X}$ is denoted by $\mathbf{X}^{\dagger}$. The commutator of two operators $\mathbf{X},\mathbf{Y}$ is defined by $\left[\mathbf{X},\mathbf{Y}\right]=\mathbf{X}\mathbf{Y}-\mathbf{Y}\mathbf{X}$. The anticommutator of two operators $\mathbf{X},\mathbf{Y}$ is defined by $\left\{\mathbf{X},\mathbf{Y}\right\}=\mathbf{X}\mathbf{Y}+\mathbf{Y}\mathbf{X}$. The tensor product of operators $\mathbf{X},\mathbf{Y}$ defined on Hilbert spaces $\mathbb{H},\mathbb{G}$ is denoted $\mathbf{X} \otimes \mathbf{Y}$, and is defined on the tensor product Hilbert space $\mathbb{H}\otimes\mathbb{G}$.

	\section{Reservoir qubits input to a harmonic oscillator}
	\label{sec:RQIO}
	
	In the system setting shown in Fig. \ref{fig:GREng},  each qubit interacts with the cavity for a fixed time $t_r$ according to the Jaynes-Cummings Hamiltonian
	\begin{align}
		\mathbf{H}_{JC} = i\frac{\Omega}{2}( | g \rangle \langle e |\mathbf{a}^{\dagger} - | e \rangle \langle g |\mathbf{a})
		\label{eq:ResInt}
	\end{align}
	with $\Omega$ the effective qubit-oscillator coupling strength (Rabi oscillation frequency), $\mathbf{a}$ the oscillator mode's annihilation operator and $| g \rangle, | e \rangle$ the qubit's ground and excited states. The unitary propagator describing one qubit-oscillator interaction is then
	\begin{align}\label{eq:ResInt}
		\mathbf{U}_r = | g \rangle \langle g |\cos\theta_{\mathbf{N}} +  | e \rangle \langle e |\cos\theta_{\mathbf{N+I}}
		- | e \rangle \langle g | \mathbf{a}\frac{\sin\theta_{\mathbf{N}}}{\sqrt{\mathbf{N}}} + | g \rangle \langle e |\frac{\sin\theta_{\mathbf{N}}}{\sqrt{\mathbf{N}}}\mathbf{a}^{\dagger}
	\end{align}
	where
	\begin{align*}
		\theta_{\mathbf{N}} = \theta\sqrt{\mathbf{N}}= \tfrac{1}{2}\Omega t_r \sum_n\sqrt{n}| n \rangle \langle n |, \;
	\end{align*}
	with $\mathbf{N} = \mathbf{a}^\dagger \mathbf{a}$ the photon number operator, $|n\rangle (n=0,1,2,...)$ the Fock states of the harmonic oscillator mode, and $\mathbf{I}$ the identity operator. Please note that our qubit-oscillator system is operating in the weakly coupled regime, and thus $\theta = \tfrac{1}{2}\Omega t_r$ is sufficiently small.
	
	
	Specifically, if the reservoir qubits are initialised at different states, we begin with the study of the qubits sent to the harmonic oscillator in pairs. The two qubits constituting a pair interact sequentially with the oscillator according to \eqref{eq:ResInt}. In order to obtain a Markovian evolution for the oscillator state, we have to keep track of the result of its interaction with the qubit pair.
	One pair is thus regarded as one effective auxiliary system, which undergoes two consecutive Hamiltonian interactions with the oscillator (first Hamiltonian coupling with the subspace corresponding to first qubit, and then with the subspace corresponding to second qubit), and the corresponding propagator is obtained as a straightforward extension of $\mathbf{U_r}$.
	In such a fashion, the initial state of one qubit pair can be written as 
	\begin{align}
		|\psi_{q^2}\rangle(0) = \beta_{gg}|gg\rangle + \beta_{ge}|ge\rangle + \beta_{eg}|eg\rangle + \beta_{ee}|ee\rangle,
		\label{eq:qubitis} 
	\end{align}
	with $\beta_{gg}, \beta_{ge}, \beta_{eg}, \beta_{ee} \in \mathbb{C}$, and $\left|\beta_{gg}\right|^2+\left|\beta_{ge} \right|^2+\left|\beta_{eg} \right|^2+\left|\beta_{ee} \right|^2=1$. The initial state of the oscillator is denoted by $\rho_c(0)$, which is arbitrary.
	The evolution of the oscillator state over these two consecutive interactions can then be described by the Kraus map:
	\begin{align}
		\rho_c(t+1) = \mathbf{M}_{gg} \rho_c(t) \mathbf{M}_{gg}^\dagger + \mathbf{M}_{ge} \rho_c(t) \mathbf{M}_{ge}^\dagger + \mathbf{M}_{eg} \rho_c(t) \mathbf{M}_{eg}^\dagger + \mathbf{M}_{ee} \rho_c(t) \mathbf{M}_{ee}^\dagger .
		\label{eq:rhodyns} 
	\end{align}
	In the basis $ (|g\rangle,|e\rangle)$ for both qubits, the associated operators read:
	\begin{align}
		\mathbf{M}_{gg} &=  \beta_{gg}\cos^{2}\theta_{\mathbf{N}}+\beta_{ge}\cos\theta_{\mathbf{N}}\frac{\sin\theta_{\mathbf{N}}}{\sqrt{\mathbf{N}}}\mathbf{a}^{\dagger}
		+\beta_{eg}\frac{\sin\theta_{\mathbf{N}}}{\sqrt{\mathbf{N}}}\mathbf{a}^{\dagger}\cos\theta_{\mathbf{N}} +\beta_{ee}\frac{\sin\theta_{\mathbf{N}}}{\sqrt{\mathbf{N}}}\mathbf{a}^{\dagger}\frac{\sin\theta_{\mathbf{N}}}{\sqrt{\mathbf{N}}}\mathbf{a}^{\dagger},\label{eq:M2qubits} \nonumber\\
		\mathbf{M}_{ge} &= -\beta_{gg}\cos\theta_{\mathbf{N}}\mathbf{a}\frac{\sin\theta_{\mathbf{N}}}{\sqrt{\mathbf{N}}} +\beta_{ge}\cos\theta_{\mathbf{N}}\cos\theta_{\mathbf{N+I}} -\beta_{eg}\sin^2\theta_{\mathbf{N}} + \beta_{ee}\frac{\sin\theta_{\mathbf{N}}}{\sqrt{\mathbf{N}}}\mathbf{a}^{\dagger}\cos\theta_{\mathbf{N+I}},\nonumber\\
		\mathbf{M}_{eg} &= -\beta_{gg}\mathbf{a}\frac{\sin\theta_{\mathbf{N}}}{\sqrt{\mathbf{N}}}\cos\theta_{\mathbf{N}}
		-\beta_{ge}\sin^2\theta_{\mathbf{N+I}}+\beta_{eg}\cos\theta_{\mathbf{N+I}}\cos\theta_{\mathbf{N}}+\beta_{ee}\cos\theta_{\mathbf{N+I}}\frac{\sin\theta_{\mathbf{N}}}{\sqrt{\mathbf{N}}}\mathbf{a}^{\dagger}, \nonumber \\
		\mathbf{M}_{ee} &= \beta_{gg}\mathbf{a}\frac{\sin\theta_{\mathbf{N}}}{\sqrt{\mathbf{N}}}\mathbf{a}\frac{\sin\theta_{\mathbf{N}}}{\sqrt{\mathbf{N}}}-\beta_{ge}\mathbf{a}\frac{\sin\theta_{\mathbf{N}}}{\sqrt{\mathbf{N}}}\cos\theta_{\mathbf{N+I}}-\beta_{eg}\cos\theta_{\mathbf{N+I}}\mathbf{a}\frac{\sin\theta_{\mathbf{N}}}{\sqrt{\mathbf{N}}}+\beta_{ee}\cos^{2}\theta_{\mathbf{N+I}}.
	\end{align}
	
	We recall the approximate Lindblad master equation in \cite{MS17} characterising the dynamics of this system setting with $\theta$ sufficiently small.  We expand this Kraus map to the second order in $\theta$, and we observe that it appears to be the discretisation of a Lindblad master equation with three dissipation channels:
	\begin{align}
		\tfrac{d}{d\tau}\rho_c(\tau) =-i\left[\mathbf{H},\rho_c(\tau)\right]+\sum_{j=1}^3\mathcal{L}\left(\mathbf{L}_{j}\right)\rho_c\left(\tau\right),
		\label{eq:LindM}
	\end{align}
	with the dissipation super-operator $\mathcal{L}(\mathbf{L})\rho_c=\mathbf{L}\rho_c\mathbf{L}^{\dagger}-\frac{1}{2}(\mathbf{L}^{\dagger}\mathbf{L}\rho_c+\rho_c \mathbf{L}^{\dagger}\mathbf{L})$, and operators
	\begin{align}
		\mathbf{H} &= -i\theta\left(\mathbf{Q}-\mathbf{Q}^{\dagger}\right),\nonumber\\
		\mathbf{Q} &=\left[\beta_{gg}\left(\beta_{ge}^{\ast}+\beta_{eg}^{\ast}\right)+\beta_{ee}^{\ast}\left(\beta_{ge}+\beta_{eg}\right)\right]\mathbf{a},\nonumber\\
		\mathbf{L}_1 &= \sqrt{2}\theta\left(\beta_{gg}\mathbf{a}  - \beta_{ee}\mathbf{a}^\dagger\right),\nonumber\\
		\mathbf{L}_2 &= \theta\left(\beta_{ge}+\beta_{eg}\right)\mathbf{a},\nonumber\\
		\mathbf{L}_3 &= \theta\left(\beta_{ge}+\beta_{eg}\right)\mathbf{a}^\dagger \; .
		\label{eq:2appie}
	\end{align}
	The a-dimensional time $\tau$ corresponds to the duration of the interaction with one pair of qubits. The decoherence operators $\mathbf{L}_2$ and $\mathbf{L}_3$ describe a purely thermal bath at infinite temperature; this simply has the effect of stabilising a high-energy thermal mixture of coherent states. However, by taking $\left|\beta_{ge} + \beta_{eg}\right|$ sufficiently small, we can make the coupling to this thermal reservoir negligible leaving the dominant terms $\mathbf{L}_1$ and $\mathbf{H}$. In this regime, the reservoir qubits can be prepared to stabilise a minimum-uncertainty squeezed state \cite{MS17}
	\begin{align*}
		|\alpha,\zeta=r e^{i\phi_r}\rangle & =\mathbf{D}\left(\alpha\right) \mathbf{S}\left(\zeta\right) |0\rangle \quad \text{where}\\
		& \mathbf{D}\left(\alpha\right)=\exp\left(\alpha \mathbf{a}^{\dagger}-\alpha^{\ast}\mathbf{a}\right) \; ,\\
		& \mathbf{S}\left(\zeta\right)=\exp\left(\tfrac{1}{2}(\zeta^{\ast}\mathbf{a}^2-\zeta(\mathbf{a}^{\dagger})^{2})\right)
	\end{align*}
	are respectively the displacement operator by $\alpha = |\alpha|e^{i\phi_\alpha} \in \mathbb{C}$ and the squeezing of the vacuum by $|\zeta|=r$ in the direction characterised by the angle $\phi_r$. That is, denoting $\mathbf{X_\phi} = \frac{\mathbf{a}e^{i\phi}+\mathbf{a}^\dagger e^{-i\phi}}{2}$ the oscillator quadrature in direction $\phi$ (we follow the conventional forms in \cite{HR06B}), the corresponding variances satisfy
	\begin{align*}
		& (\Delta\mathbf{X}_{\frac{\phi_r}{2}}) (\Delta\mathbf{X}_{\frac{\phi_r+\pi}{2}}) = \frac{1}{4}, \\
		& (\Delta\mathbf{X}_{\frac{\phi_r}{2}}) = \tfrac{1}{2}\, e^{-r}
	\end{align*}
	for $|\psi\rangle = \mathbf{S}\left(r\right) |0\rangle$. Therefore, such states saturate the Heisenberg uncertainty inequality, with less uncertainty on $\mathbf{X}_{\phi_r/2}$ than a classical-like state such as the vacuum $|0\rangle$.

	\section{The necessity of having entanglement in the reservoir input qubits}
	\label{sec:NIE}
	
	\subsection{Stabilisation of squeezed states: entangled input vs. separable input}
	
	Concerning the stabilisation problem, in \cite{MS17} we focus on entangled reservoir input qubits, as regard to the stabilisation of highly squeezed states of the oscillator instead of coherent states. In this section, we will further explore why we are in need of entanglement in the reservoir input qubits considered in pairs. First, we notice that time-varying reservoir input qubits can also stabilise squeezed states of the oscillator.
	
	\begin{figure}[!htp]
		\centering
		\includegraphics[scale=.6]{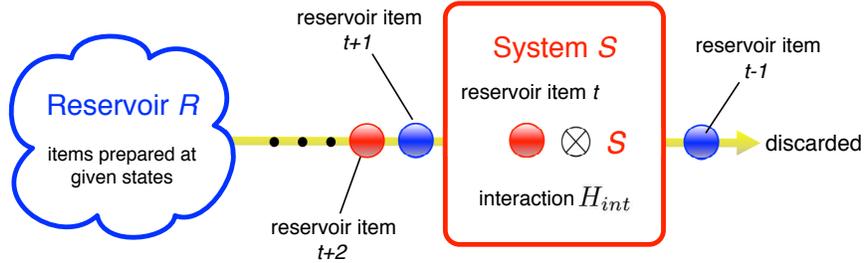}
		\caption{Framework of open-loop quantum reservoir engineering with qubits input considered in pairs. Entanglement may be involved in a qubit pair. Specifically, in the time-varying separable input case, the first and second qubits are prepared at different states, say $|\psi_q,+\rangle$ and $|\psi_q,-\rangle$. Then the third and fourth qubits are initialised at $|\psi_q,+\rangle$ and $|\psi_q,-\rangle$ respectively, et cetera.}
		\label{fig:GREngpair}
	\end{figure}
	
	From \cite{SRBR11} we are convinced that if each qubit is prepared at the same state $|\psi_q \rangle = \cos u |g\rangle + e^{i\chi}\sin u  |e\rangle$, with $u\in(0,\pi/4)$ and $\chi \in [-\pi,\pi)$,  the oscillator will be stabilised at the coherent state $|\psi_c \rangle = |\alpha \rangle = \frac{2u}{\theta}e^{i\chi}$. By contrast, if we initialise one pair of qubits in the reservoir at the states $\cos u |g\rangle + \sin u |e\rangle$ and $\cos u |g\rangle - \sin u |e\rangle$ alternately with $u\in (0,\pi/4)$, as shown in Fig. \ref{fig:GREngpair}, we find that the oscillator will be stabilised at the squeezed state $|0,re^{i\pi}\rangle$ where $\tanh r =  \tan^2 u$. And the corresponding convergence rate is $\kappa  = 2\theta^2 \cos 2u$. However, in this simple example, only squeezed vacuum states can be stabilised. We are now interested to see whether a larger set of squeezed states can be stabilised using entangled input than separable input. Before moving to the detailed proof for the necessity of having entanglement in the input, we provide a generalisation of Theorem 1 in \cite{MS17}, in terms of stabilising the oscillator at squeezed states.  
	
	With the aim of stabilising the oscillator at a steady state, it is not surprising to prepare each pair of qubits  at an identical joint state. As discussed in  \ref{sec:RQIO}, we require that  $
	\beta_{ge}+\beta_{eg} = \epsilon (\left|\epsilon \right| \ll 1)$, and $\left|\beta_{ee}\right| < \left|\beta_{gg}\right| $ to ensure convergence of the oscillator's state  \cite{MS17}.
	\begin{theorem}
		Consider the Lindblad master equation
		\begin{align}
			\tfrac{d}{dt}\rho_c(\tau) =-i\left[\mathbf{H},\rho_c(\tau)\right]+\mathcal{L}\left(\mathbf{L}_{1}\right)\rho_c\left(\tau\right)
			\label{eq:4thm1}
		\end{align}
		which describes, according to approximations just discussed, the engineered reservoir obtained through resonant interaction of a harmonic oscillator with a stream of consecutive qubit pairs initialised in the state \eqref{eq:qubitis} before interaction. This Lindblad master equation stabilises the squeezed state $|\alpha,\zeta = re^{i\phi_r} \rangle$ provided we initialise the qubit pairs as
		\begin{align*}
			& \beta_{ge}+\beta_{eg} = \epsilon \;(\left|\epsilon \right| \ll 1) ,\\
			& \left|\beta_{ee}\right| < \left|\beta_{gg}\right|.
		\end{align*}
		with the parameters tuned as:
		\begin{align}
			\tanh(-r)&=\frac{\left|\beta_{ee}\right|}{\left|\beta_{gg}\right|}, \label{eq:findr}\\
			\phi_{r}&=\phi_{ee}-\phi_{gg} \; \left(\mathrm{mod} \; 2\pi\right),\label{eq:findpr} \\
			\alpha&=\frac{\epsilon\beta_{gg}^{\ast}+\epsilon^{\ast}\beta_{ee}}{\theta\left(\left|\beta_{gg}\right|^{2}-\left|\beta_{ee}\right|^{2}\right)}.\label{eq:finda}
		\end{align}
		The convergence rate towards $|\alpha,\zeta = re^{i\phi_r} \rangle$ is
		\begin{align}
			\kappa = 2\theta^2\left(\left|\beta_{gg}\right|^{2}-\left|\beta_{ee}\right|^{2}\right).
			\label{eq:findka}
		\end{align}
		\label{thm:sta}
	\end{theorem}
	
	\begin{proof}
		Apply the transformation
		\begin{align}
			\tilde{\rho}_c (t)= \mathbf{S}^\dagger(\zeta)\mathbf{D}^\dagger(\alpha)\rho_c(t) \mathbf{D}(\alpha)\mathbf{S}(\zeta).
			\label{eq:Rhotran}
		\end{align}
		By using the properties
		\begin{align*}
			\mathbf{D}^\dagger(\alpha)\mathbf{a}\mathbf{D}(\alpha) &= \mathbf{a} + \alpha,\nonumber\\
			\mathbf{S}^\dagger(re^{i\phi_r})\mathbf{a}\mathbf{S}(re^{i\phi_r}) &= \mathbf{a}\cosh r - e^{i\phi_r}\mathbf{a}^\dagger \sinh r, 
		\end{align*}
		and applying the transformation
		\begin{align}
			\tilde{\rho}_c (t)= \mathbf{S}^\dagger(\zeta)\mathbf{D}^\dagger(\alpha)\rho_c(t) \mathbf{D}(\alpha)\mathbf{S}(\zeta),
			\label{eq:Rhotran}
		\end{align}
		we obtain that the corresponding Lindblad master equation for $\tilde{\rho}_c (t)$ is dominated by the following Hamiltonian and coupling operators
		\begin{align}
			\tilde{H} &= -i\theta \times \nonumber\\
			&\left\{ \left[\left(\epsilon^{\ast}\beta_{gg}+\epsilon\beta_{ee}^{\ast}\right)\cosh r+\left(\epsilon\beta_{gg}^{\ast}+\epsilon^{\ast}\beta_{ee}\right)e^{-i\phi_{r}}\sinh r\right]\mathbf{a}\right.\nonumber\\
			&\left.-\left[\left(\epsilon^{\ast}\beta_{gg}+\epsilon\beta_{ee}^{\ast}\right)e^{i\phi_{r}}\sinh r+\left(\epsilon\beta_{gg}^{\ast}+\epsilon^{\ast}\beta_{ee}\right)\cosh r\right]\mathbf{a}^{\dagger}\right.\nonumber\\
			&\left.+\left(\epsilon^{\ast}\beta_{gg}+\epsilon\beta_{ee}^{\ast}\right)\alpha+\left(\epsilon\beta_{gg}^{\ast}-\epsilon^{\ast}\beta_{ee}\right)\alpha^{\ast}\right\},\nonumber\\
			\tilde{L}_{1}&=\sqrt{2}\theta\left[\left(\beta_{gg}\cosh r+\beta_{ee}e^{-i\phi_{r}}\sinh r\right)\mathbf{a}-\left(\beta_{gg}e^{i\phi_{r}}\sinh r+\beta_{ee}\cosh r\right)\mathbf{a}^{\dagger}+\beta_{gg}\alpha-\beta_{ee}\alpha^{\ast}\right]. \nonumber
		\end{align}
		Plugging in the initial parameters for qubit pairs provided in the statement, one can have
		\begin{align}
			\tfrac{d}{d\tau}\tilde{\rho}_c\left(\tau\right) = \kappa\; \mathcal{L}(\mathbf{a}) \tilde{\rho}_c(\tau),
			\label{eq:LindMvac}
		\end{align}
		where
		\begin{align*}
			\kappa = 2\theta^2\left(\left|\beta_{gg}\right|^{2}-\left|\beta_{ee}\right|^{2}\right).
		\end{align*}
		This equation stabilises $\tilde{\rho}_c$ towards the vacuum state at the rate $\kappa$. The converse change of variables yields that the oscillator is stabilised at the state $|\alpha,\zeta = re^{i\phi_r} \rangle$ as detailed in the statement.
		
	\end{proof}

	In Theorem \ref{thm:sta}, entanglement in the reservoir input qubits is not required. And please note that if and only if $\beta_{gg}\beta_{ee}=\beta_{ge}\beta_{eg}$, there is no entanglement in the reservoir input qubits.

	The following theorem illustrates the advantage and necessity of introducing entanglement in the reservoir input.
	\begin{theorem}
		\label{thm:stedis}
		To stabilise the oscillator at a squeezed state $|\alpha,\zeta = re^{i\phi_r} \rangle$ given in Theorem \ref{thm:sta}, using the specific setup depicted in Fig. \ref{fig:GREngpair}, entangled input enables us to stabilise a strictly larger set of squeezed states than separable input. In more concrete terms, with a given $\theta$,
		the same $ \epsilon = \beta_{ge}+\beta_{eg}$, and $\mu = \frac{|\beta_{gg}|}{|\beta_{ee}|}$, we have that
		\begin{align*}
			\max|\alpha|_{ent} \geq \frac{\left|\epsilon\right|}{\theta\sqrt{1-\left|\epsilon\right|^2}\left(\mu-1\right)} >  \max|\alpha|_{sep},
		\end{align*}
		where $\max|\alpha|_{ent}$ and $\max|\alpha|_{sep}$ denote the largest amplitudes that can be achieved using entangled and separable input respectively.
	\end{theorem}

	\begin{proof}
		According to Theorem \ref{thm:sta}, we know $\mu = \frac{|\beta_{gg}|}{|\beta_{ee}|}> 1$. It is not difficult to obtain that
		\begin{align*}
			|\alpha|^2 = \frac{\left|\epsilon\right|^{2}\left(1+\mu^{2}+2\mu\cos\left(\phi_{gg}+\phi_{ee}-2\phi_{\epsilon}\right)\right)}{\theta^{2}\left(\mu^{2}-1\right)^{2}\left|\beta_{ee}\right|^{2}}.
		\end{align*}
		In the separable case, we further have
		\begin{align*}
			|\epsilon|^2 = 1 - \left(\mu^2+1\right)|\beta_{ee}|^2 + 2\mu\cos\left(\phi_{ge}-\phi_{eg}\right)|\beta_{ee}|^2,
		\end{align*}
		which is equivalent to
		\begin{align*}
			|\beta_{ee}|^2 = \frac{1-|\epsilon|^2}{1+\mu^2-2\mu\cos\left(\phi_{ge}-\phi_{eg}\right)},
		\end{align*}
		where we make use of the fact
		\begin{align*}
			|\beta_{gg}||\beta_{ee}|=|\beta_{ge}||\beta_{eg}|.
		\end{align*}
		Therefore, the amplitude of the steady state can be rewritten as a function of $|\epsilon|$, $\phi_{gg}$, $\phi_{ee}$, $\phi_{eg}$ and $\phi_{ge}$, i.e.,
		\begin{align*}
			|\alpha|^2_{sep} = \frac{\left|\epsilon\right|^{2}}{1-\left|\epsilon\right|^{2}}\frac{\left(1+\mu^{2}+2\mu\cos\left(\phi_{gg}+\phi_{ee}-2\phi_{\epsilon}\right)\right)}{\theta\left(\mu^{2}-1\right)} \frac{\left(1+\mu^{2}-2\mu\cos\left(\phi_{ge}-\phi_{eg}\right)\right)}{\theta\left(\mu^{2}-1\right)}.
		\end{align*}
		For any $|\epsilon| \ll 1$, in order to obtain the maximal $|\alpha|$, one has to require that
		\begin{align}
			\phi_{ge} - \phi_{eg} &= \pi \; (\mathrm{mod} \; 2\pi),\label{eq:sepmax1}\\
			\phi_{gg} + \phi_{ee} - 2\phi_{\epsilon} &= 0  \; (\mathrm{mod} \; 2\pi).\label{eq:sepmax2}
		\end{align}
		Moreover, because
		$\tan \phi_{\epsilon} = \frac{\left|\beta_{ge}\right|\sin\phi_{ge}+\left|\beta_{eg}\right|\sin\phi_{eg}}{\left|\beta_{ge}\right|\cos\phi_{ge}+\left|\beta_{eg}\right|\cos\phi_{eg}}$,
		one can conclude that 
		\begin{align}
			\phi_{\epsilon} = \phi_{eg}  \;(\mathrm{mod} \; 2\pi).
			\label{eq:sepep}
		\end{align}
		
		However, in the separable case, it must be satisfied that
		\begin{align}
			\phi_{gg} + \phi_{ee} - \phi_{ge} - \phi_{eg} = 0 \;(\mathrm{mod} \; 2\pi),
			\label{eq:sepph}
		\end{align}
		and thus from equations  \eqref{eq:sepmax1},\eqref{eq:sepep},\eqref{eq:sepph},  we have
		\begin{align}
			\phi_{gg} + \phi_{ee} - 2\phi_{\epsilon} = \phi_{ge} + \phi_{eg} - 2\phi_{\epsilon} = \pi \;(\mathrm{mod} \; 2\pi),
		\end{align}
		which contradicts with the condition in equation~\eqref{eq:sepmax2}.
		
		Hence, in the absence of entanglement in the reservoir input, it is obvious that
		\begin{align}
			\max|\alpha|_{sep} < \frac{\left|\epsilon\right|}{\theta\sqrt{1-\left|\epsilon\right|^2}\left(\mu-1\right)}.
		\end{align}
		
		By contrast, if we allow for entanglement in the reservoir input qubits, there is no such a constraint that $\phi_{gg} + \phi_{ee} - \phi_{ge} - \phi_{eg} = 0 (\mod 2\pi)$. Therefore, even if we still choose $|\beta_{gg}||\beta_{ee}|=|\beta_{ge}||\beta_{eg}|$, we can make $\phi_{gg} + \phi_{ee} - 2\phi_{\epsilon} = 0 (\mod 2\pi)$ hold in the entangled case. That is to say,
		\begin{align*}
			\max|\alpha|_{ent} = \frac{\left|\epsilon\right|}{\theta\sqrt{1-\left|\epsilon\right|^2}\left(\mu-1\right)}
		\end{align*}
		under the conditions mentioned above.
		
		Furthermore, if we choose $|\beta_{gg}||\beta_{ee}|-|\beta_{ge}||\beta_{eg}|<0$, it can be easily verified that
		\begin{align*}
			\max|\alpha|_{ent} > \frac{\left|\epsilon\right|}{\theta\sqrt{1-\left|\epsilon\right|^2}\left(\mu-1\right)}.
		\end{align*}
		It can thus be concluded that
		\begin{align}
			\max|\alpha|_{ent} \geq \frac{\left|\epsilon\right|}{\theta\sqrt{1-\left|\epsilon\right|^2}\left(\mu-1\right)}.
		\end{align}
		One can further define $\left|\epsilon^o\right|=\max_{\left|\epsilon\right|}\left\{ \epsilon = \beta_{ge}+\beta_{eg}\mid \right.$
		the system (6) can be stabilised at a squeezed state$\left.\right\}$.
		The maximal amplitudes that are reachable in the entangled input case, denoted by $|\alpha^o|_{ent}$, and in the separable input case, denoted by  $|\alpha^o|_{sep}$, satisfy the following inequality
		\begin{align*}
			|\alpha^o|_{ent} \geq \frac{\left|\epsilon^o\right|}{\theta\sqrt{1-\left|\epsilon^o\right|^2}\left(\mu-1\right)}>|\alpha^o|_{sep}.
		\end{align*}

	\end{proof}
	
	From Theorem \ref{thm:stedis} we know $\theta$ and $\mu$ determine the convergence rate, and the squeezing strength is totally determined by $\mu$. With a strong squeezing strength ($\mu\rightarrow1^{+}$), both theoretically and practically $\theta$ should not be overly small. In fact, it is thus reasonable to bound $\theta$ below as $\min \theta = \theta^o$. In Fig. \ref{fig:EntLAmp}, we illustrate Theorem \ref{thm:stedis} by comparing the amplitude of a squeezed state stabilised by separable input to that stabilised by entangled input in terms of  Wigner quasi-probability distribution. 
	\begin{figure}[!htp]
		\centering
		\includegraphics[scale=.6]{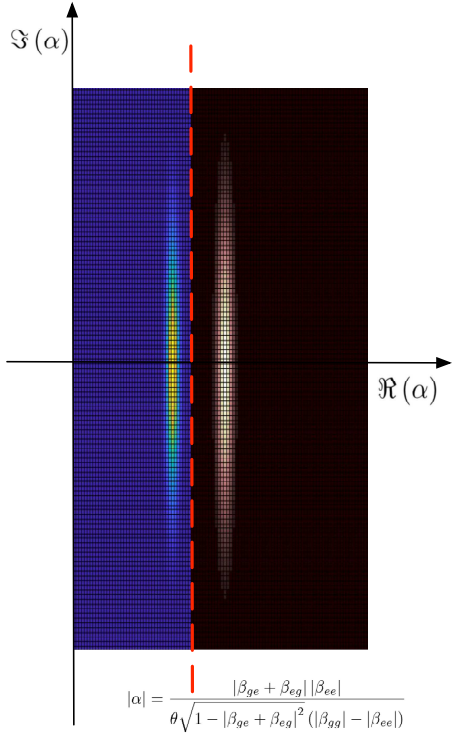}
		\caption{A comparison between the amplitudes of squeezed states (Wigner quasi-probability distribution) stabilised by entangled and separable reservoir input, with the same squeezing strength and given $\theta$. The coordinates are located along the squeezing direction and shifted such that $\Im(\alpha) = 0$. 
			The dashed line corresponds to the boundary that a squeezed state stabilised by separable qubits cannot reach. Only by entangling the reservoir input qubits can we stabilise a squeezed state on the right hand side of the dashed line.}
		\label{fig:EntLAmp}
	\end{figure}

\subsection{Estimation of the parameter $\theta$: entangled input vs. separable input}
Inspired by the work in \cite{EMD11,ARTG18}, our system setup shown in Fig. \ref{fig:GREngpair} is a very good candidate for parameter estimation (e.g. coupling strength estimation). However, in this setting, we do not have to impose measurement on the qubits which is intrinsically different from that in \cite{ARTG18}. Furthermore, there is no requirement for the initial state of the oscillator. As stated in Theorem \ref{thm:stedis}, the oscillator will be regulated to a squeezed state $\rho_c = |\alpha,re^{i\phi_r} \rangle  \langle\alpha,re^{i\phi_r}|$ in the weakly coupled regime (namely $\theta = \tfrac{1}{2}\Omega t_r$ is sufficiently small). Now in particular we would like to estimate the coupling strength characterised by $\theta$, as an application of our stabilisation result in quantum metrology. In the following part we write the steady state $\rho_c = |\alpha,re^{i\phi_r} \rangle  \langle\alpha,re^{i\phi_r}|$ as $\rho_c(\theta)$.
	
The local precision limit of estimating $\theta$ from the output state $\rho_c(\theta)$ is related to the Bures distance between two neighbouring quantum states $\rho_c(\theta)$ and $\rho_c(\theta+d\theta)$ as \cite{BC94B}
\begin{eqnarray}
		\label{eq:Bures}
		d^2_{Bures}[\rho_c(\theta),\rho_c(\theta+d\theta)]=\frac{1}{4}J_Qd\theta^2,
\end{eqnarray}
where $d_{Bures}[\rho_1,\rho_2]=\sqrt{2-2F(\rho_1,\rho_2)}$ and $F(\rho_1,\rho_2)=\mathrm{tr} \sqrt{\rho_1^{\frac{1}{2}}\rho_2\rho_1^{\frac{1}{2}}}$ is the fidelity between two states. Here $J_Q$ denotes the QFI.
	
The following theorem highlights the role of entanglement in the reservoir input qubits, by calculating the difference between the values of QFI in both separable and entangled input cases. 
\begin{theorem}
\label{thm:qfi}
For a steady squeezed state $\rho_c(\theta)$ given in Theorem \ref{thm:sta}, a strictly higher estimation precision with respect to the parameter $\theta$ can be achieved by entangled reservoir input qubits. In terms of quantum Fisher information (QFI) $\mathcal{J}$, we have that
		\begin{align*}
			\max\mathcal{J}_{sep} < &\frac{16\left|\epsilon^{2}\right|\left(\mu+1\right)^{3}}{\left(\mu-1\right)^{3}} \leq \max\mathcal{J}_{ent},
		\end{align*}
		with the parameters defined in Theorems  \ref{thm:sta} and \ref{thm:stedis}.
	\end{theorem}
	\begin{proof}
		For squeezed states, we can calculate the QFI $J_Q$ in the method given as follows \cite{Helstrom76B}:
		\begin{align}
			J_Q = 2\dot{d}^\dagger\sigma^{-1}\dot{d},
			\label{eq:Jqdd}
		\end{align}
		where
		\begin{align}
			d &= \left(\alpha,\alpha^\ast \right)^T,\\
			\sigma &= \left[\begin{array}{cc}
				\cosh2r & -e^{i\phi_{r}}\sinh2r\\
				-e^{-i\phi_{r}}\sinh2r & \cosh2r
			\end{array}\right],
		\end{align}
		with all the parameters given in Theorem \ref{thm:sta}. Therefore,  in this setup $J_Q$  can be explicitly written as 
		\begin{align}
			J_Q = 4|\dot{\alpha}|^2\left(\cos2\phi_{\alpha}\cosh2r + \cos \phi_r \sinh 2r\right).
			\label{eq:JQ}
		\end{align}
		Note that $\dot{d}$ and $\dot{\alpha}$ in equations \eqref{eq:Jqdd} and \eqref{eq:JQ} denote the derivatives with respect to $\theta$, and equivalently
		\begin{align*}
			J_Q = 4\frac{\left|\epsilon\beta_{gg}^{\ast}+\epsilon^{\ast}\beta_{ee}\right|^{2}}{\theta^{4}\left(\left|\beta_{gg}\right|^{2}-\left|\beta_{ee}\right|^{2}\right)^{2}}\frac{\left(1+\mu^{2}\right)\cos2\phi_{\alpha}-2\mu\cos\phi_{r}}{\mu^{2}-1}.
		\end{align*}
		Because we are dealing with the steady state, the time cost should be included to evaluate a protocol \cite{BC94B}. As a result, we focus on $\mathcal{J} :=  J_Q\kappa^2$, that is,
		\begin{align}
			\mathcal{J} = 16\left|\epsilon\right|^{2}\left(1+\mu^{2}+2\mu\cos\left(\phi_{gg}+\phi_{ee}-\phi_{\epsilon}\right)\right) \left|\beta_{ee}\right|^{2}\frac{\left(1+\mu^{2}\right)\cos2\phi_{\alpha}-2\mu\cos\phi_{r}}{\mu^{2}-1}.
			\label{eq:JQd}
		\end{align}
		Here, we know that
		\begin{align}
			\tan \phi_{\alpha} &= \frac{\mu\sin\left(\phi_{\epsilon}-\phi_{\phi_{gg}}\right)+\sin\left(\phi_{ee}-\phi_{\epsilon}\right)}{\mu\cos\left(\phi_{\epsilon}-\phi_{\phi_{gg}}\right)+\cos\left(\phi_{ee}-\phi_{\epsilon}\right)},\nonumber\\
			\phi_{r}&=\phi_{ee}-\phi_{gg} \; \left(\mathrm{mod} \; 2\pi\right).
			\label{eq:Jqan}
		\end{align}
		In the separable case, due to the fact $\beta_{gg}\beta_{ee}=\beta_{ge}\beta_{eg}$, one can simplify equation \eqref{eq:JQd} as
		\begin{align}
			\mathcal{J}_{sep} = \frac{16\left|\epsilon\right|^{2}\left(1-\left|\epsilon\right|^{2}\right)}{1+{\mu}^{2}-2\mu\cos\left(\phi_{ge}-\phi_{eg}\right)}\frac{\left(1+\mu^{2}\right)\cos2\phi_{\alpha}-2\mu\cos\phi_{r}}{\mu^{2}-1}.
		\end{align}
		In order to obtain the largest $\mathcal{J}_{sep}$, we have to require that
		\begin{align}
			\phi_{ee} - \phi_{gg} &= \pi \;(\mathrm{mod} \; 2\pi), \;
			\phi_{ge} - \phi_{eg} = 0  \;(\mathrm{mod} \; 2\pi),\nonumber\\
			\phi_{\alpha} &= 0 \; (\mathrm{mod} \; \pi),
		\end{align}
		which combining with the condition in equation \eqref{eq:Jqan} further indicates
		\begin{align}
			\phi_{\epsilon} = \phi_{ge} = \phi_{ee} \;(\mathrm{mod} \; 2\pi).
		\end{align}
		However, this contradicts the fact $\phi_{gg} + \phi_{ee} - \phi_{ge} - \phi_{eg} = 0 \;(\mathrm{mod} \; 2\pi)$ for separable reservoir qubits input. Hence, one can conclude that
		\begin{align}
			\max\mathcal{J}_{sep} < \frac{16\left|\epsilon^{2}\right|\left(\mu+1\right)^{3}}{\left(\mu-1\right)^{3}}.
		\end{align}
		By contrast, if we allow for entanglement in the reservoir input, there is no such a constraint that $\phi_{gg} + \phi_{ee} - \phi_{ge} - \phi_{eg} = 0 \;(\mathrm{mod} \; 2\pi)$. Therefore, even if we still choose $|\beta_{gg}||\beta_{ee}|=|\beta_{ge}||\beta_{eg}|$, we can achieve
		\begin{align*}
			\max\mathcal{J}_{ent} = \frac{16\left|\epsilon^{2}\right|\left(\mu+1\right)^{3}}{\left(\mu-1\right)^{3}}.
		\end{align*}
		Furthermore, if we take $|\beta_{gg}||\beta_{ee}|-|\beta_{ge}||\beta_{eg}|<0$, it is then straightforward that
		\begin{align*}
			\max\mathcal{J}_{ent} > \frac{16\left|\epsilon^{2}\right|\left(\mu+1\right)^{3}}{\left(\mu-1\right)^{3}}.
		\end{align*}
		Therefore,
		\begin{align}
			\max\mathcal{J}_{ent} \geq \frac{16\left|\epsilon^{2}\right|\left(\mu+1\right)^{3}}{\left(\mu-1\right)^{3}}.
		\end{align}
		This thus emphasises the necessity of having entanglement in the input to improve the precision of estimating $\theta$. The estimation error is bounded below as $\delta\hat{\theta} \ge \frac{1}{\sqrt{\mathcal{J}}}$.
	\end{proof}
	It is also worth mentioning that in standard metrology framework, the associated QFI will decay due to the noisy dynamics of open quantum systems \cite{DRM14}. By contrast, having a stabilised state protects us from involving decay in the QFI as given in Theorem \ref{thm:qfi}; see \cite{MS17} for the details of practical imperfection analysis.
	
\section{Conclusions and Future work}
\label{sec:C}
The contribution of this paper is threefold. Firstly, we explore the scenario where separable reservoir qubits, initialised at different states alternately in pairs, input to the harmonic oscillator. We show that with consecutive pairs of separable ``time-varying" input qubits, one can stabilise the oscillator at a squeezed state. Secondly, compared with our results in \cite{MS17}, we prove that entanglement is essential in the reservoir in order to stabilise a larger set of squeezed coherent states (e.g. with larger amplitudes) where separable input qubits are not adequate. Last but not least, from the view of quantum metrology, we demonstrate the necessity of entanglement in the reservoir qubits, aiming to improve the estimation precision of the coupling strength between each qubit and the oscillator. In pursuit of more generalised analysis, in the future, we will consider involving entanglement in the reservoir qubits  in a more complicated manner. We would like to see what quantum states of the oscillator can be stabilised, which may be applied to quantum sensing. The role of entanglement in reservoir engineering will thus be further expounded.
	
The authors would like to thank Alain Sarlette, Pierre Rouchon and Mazyar Mirrahimi for early discussions on this project.

\end{document}